\documentclass[a4paper,USenglish,cleveref, autoref, thm-restate]{compgeomtop}
\usepackage{graphicx,amssymb,amsmath,thm-restate,tikz}
\usepackage{xspace}

\title{Approximating the Directed Hausdorff Distance}

\author{Oliver A. Chubet}{Department of Computer Science, North Carolina State University}{oachubet@ncsu.edu}{}{}

\author{Parth M. Parikh}{Department of Computer Science, North Carolina State University}{pmparikh@ncsu.edu}{}{}

\author{Donald R. Sheehy}{Department of Computer Science, North Carolina State University}{don.r.sheehy@gmail.com}{}{}

\author{Siddharth S. Sheth}{Department of Computer Science, North Carolina State University}{sheth.sid@gmail.com}{}{}

\authorrunning{O.\,A. Chubet and P.\,M. Parikh and D.\,R. Sheehy and S.\,S. Sheth}


\relatedversion{} 

\newcommand{\ball}{\mathbf{ball}}
\newcommand{\e}{\varepsilon}

\newcommand{\dist}{\mathbf{d}}
\newcommand{\spread}{\Delta}

\newcommand{\algoHD}{\textsc{Hausdorff}\xspace}
\newcommand{\algokHD}{\textsc{k-Hausdorff}\xspace}

\newcommand{\pred}{T}
\newcommand{\points}{\mathsf{pts}}
\newcommand{\nodectr}[1]{\mathsf{ctr}(#1)}
\newcommand{\noderad}[1]{\mathsf{rad}(#1)}
\newcommand{\node}[1]{\boldsymbol{#1}}
\newcommand{\lb}{\ell}

\newcommand{\ctrdist}{\dist_{\mathsf{ctr}}}
\newcommand{\kth}{k^\textrm{th}}
\newcommand{\khdist}[1]{\dist_h^{(#1)}}

\begin{document}
    \thispagestyle{empty}
    \maketitle

\begin{abstract}
    The Hausdorff distance is a metric commonly used to compute the set similarity of geometric sets.
    For sets containing a total of $n$ points, the exact distance can be computed na\"{i}vely in $O(n^2)$ time.
    In this paper, we show how to preprocess point sets individually so that the Hausdorff distance of any pair can then be approximated in linear time.
    We assume that the metric is doubling.
    The preprocessing time for each set is $O(n\log \spread)$ where $\spread$ is the ratio of the largest to smallest pairwise distances of the input.
    In theory, this can be reduced to $O(n\log n)$ time using a much more complicated algorithm.
    We compute $(1+\e)$-approximate Hausdorff distance in $(2 + \frac{1}{\e})^{O(d)}n$ time in a metric space with doubling dimension $d$.
    The $k$-partial Hausdorff distance ignores $k$ outliers to increase stability.
    Additionally, we give a linear-time algorithm to compute directed $k$-partial Hausdorff distance for all values of $k$ at once with no change to the preprocessing.
\end{abstract}



    \section{Introduction}\label{sec:intro}
The Hausdorff distance is a metric on compact subsets of a metric space.
Let $(X,\dist)$ be a metric space and let $A$ and $B$ compact subsets of $X$.
The distance from a point $x\in X$ to the set $B$ is $\dist(x,B) := \min_{b\in B}\dist(x,b)$.
The \emph{directed Hausdorff distance} is $\dist_h(A,B) := \max_{a\in A}\dist(a,B)$, and the (undirected) \emph{Hausdorff distance} is $\dist_H(A,B) := \max\{\dist_h(A,B), \dist_h(B,A)\}$.
This definition leads directly to a quadratic time algorithm for finite sets.

In our approach, we first preprocess the input sets $A$ and $B$ individually into linear-size metric trees, specifically, greedy trees~\cite{chubet23proximity}.
All the points of a set are stored as leaves in the greedy tree and an internal node approximates the leaves in its subtree by a ball.
A subset of greedy tree nodes of radius at most $\e$ such that every point of the set is a leaf of exactly one node forms an $\e$-net of the underlying set.
The greedy tree can be used to construct $\e$-nets of the underlying set at different scales.
Our algorithms use greedy trees to maintain such nets for both sets and tracks which nodes of $B$ are close to nodes of $A$.
This approach batches the searches of the na\"{i}ve algorithm and results in fewer distance computations.

If the input contains a total of $n$ points, then preprocessing takes $\left(\frac{1}{\e}\right)^{O(d)}n\log \spread$ time.
Here $\spread$ is the \emph{spread} of the input sets, and $d$ is the doubling dimension of the metric space.
We present an algorithm that computes a $(1+\e)$-approximation of the directed Hausdorff distance from $A$ to $B$ in $\left(2+\frac{1}{\e}\right)^{O(d)}n$ time after preprocessing.
The preprocessing is especially useful when the same sets are involved in multiple distance computations.

One difficulty of working with the Hausdorff distance is its sensitivity to outliers.
There are several variations of the Hausdorff distance that reduce the sensitivity to outliers.
Among the simplest is the following definition where one can ignore up to $k$ outliers.
The \emph{$k^\text{th}$-partial} directed Hausdorff distance is $\dist_h^{(k)}(A,B) := \min_{S\in A^{(k)}} \dist_h(S,B)$ where $A^{(k)}$ is the set of all subsets of $A$ with $k$ points removed~\cite{Huttenlocher93Comparing}.

One might expect that this is a harder problem than computing the directed Hausdorff distance.
However, we show that for approximations in low dimensions, this is not the case; the worst-case running time matches that of our algorithm for the standard case (up to an additive $\log\spread$ term).
We present an algorithm that computes a $(1+\e)$-approximation of $\dist_h^k(A, B)$ for all values of $k$ in $\left(2+\frac{1}{\e}\right)^{O(d)}n + O(\log \spread)$ time after the same preprocessing as before.
That is, the output is a list of $n+1$ approximate values of $\dist_h^k(A, B)$ for $k \in \{0, \ldots, n\}$ and this list is produced in linear time.



    \section{Related Work}\label{sec:related}
We focus on general metrics with bounded doubling dimension.
In the general setting, one may not expect a subquadratic algorithm for computing the Hausdorff distance, however, there are classes of metric spaces for which the Hausdorff distance can be computed more quickly.
For example, given point sets in the plane, fast nearest neighbor search data structures~\cite{Alt95Approximate} are used to give an $O(n\log n)$ time algorithm.
If one allows for approximate answers, $O(n\log n)$ time algorithms are possible in low-dimensional Euclidean spaces~\cite{Arya98Optimal}.

It is not easy to get an asymptotic improvement on the na\"{i}ve Hausdorff distance algorithm without using an efficient data structure in higher dimensions.
In practice, there exist many heuristics to speed up the na\"{i}ve algorithm~\cite{Chen17Local, Ryu21Efficient, Taha15Efficient}.
Another popular technique is to use a geometric tree data structure.
Zhang et al.~\cite{Zhang17Efficient} use octrees to compute the exact Hausdorff distance between 3D point sets.
Nutanong et al.~\cite{Nutanong11Incremental} present an algorithm to compute the exact Hausdorff distance using R-trees.

The partial Hausdorff distance was first introduced by Huttenlocher et al.~\cite{Huttenlocher93Comparing}.
Although there has been considerable interest in this pseudometric, most results are experimental and to the best of our knowledge, a theoretical running time bound does not exist.
We give an algorithm to compute approximate partial Hausdorff distance that runs in linear time after preprocessing.

The algorithms presented in this paper maintain both input sets as metric trees and traverse them simultaneously.
This approach is similar to that of Nutanong et al.~\cite{Nutanong11Incremental}, but more generally it has been studied in the machine learning literature as dual-tree algorithms~\cite{gray00n-body}.
Search problems such as the $k$-nearest-neighbor search~\cite{curtin13tree}, range search~\cite{curtin13tree}, and the all-nearest-neighbor search~\cite{ram09linear} have been explored using dual-tree algorithms.
Any all-nearest-neighbor search algorithm where the query and reference sets are different can also be used to compute the directed Hausdorff distance.
Ram et al.~\cite{ram09linear} present an all-nearest-neighbor algorithm that runs in linear time under some strict assumptions about the underlying metrics.
The running time of their algorithm depends on a constant called the degree of bichromaticity.
Moreover, this dependence is exponential in the degree of bichromaticity, and high values can result in a poor bound~\cite{curtin15plug}.



    \section{Background}\label{sec:background}
\subsection{Doubling Metrics}
Let $(X,\dist)$ be a \emph{metric space}.
A \emph{metric ball} is a subset $\ball(c, r)$ of $X$, with center $c \in X$ and radius $r \ge 0$ such that $\ball(c, r) := \{x \in X \mid \dist(x, c) \le r\}$.
The \emph{spread} $\spread$ of $A \subseteq X$ is the ratio of the diameter to the smallest pairwise distance of points in $A$.
The \emph{doubling dimension} of $X$, denoted $\dim(X)$, is the smallest real number $d$ such that any metric ball in $X$ can be covered by at most $2^d$ balls of half the radius.
If $\dim(X)$ is bounded then $X$ is a \emph{doubling metric}.
The set $A$ is \emph{$\lambda$-packed} if $d(a,b)\ge \lambda$ for any distinct $a,b\in A$.
The following lemma by Krauthgamer and Lee~\cite{krauthgamer04navigating} limits the cardinality of packed and bounded sets.
\begin{lemma}[Standard Packing Lemma]\label{lem:std_packing}
    Let $(X, \dist)$ be a metric space with $\dim(X) = d$.
    If $Z \subseteq X$ is $r$-packed and can be covered by a metric ball of radius $R$ then $|Z| \le \left(\frac{4R}{r}\right)^d$.
\end{lemma}

\subsection{Greedy Permutations and Predecessor Maps}
Let $P = (p_0,\ldots, p_{n-1})$ be a finite sequence of points in a metric space with distance function $\dist$.
The $i^{th}$-\emph{prefix} is the set $P_i = \{p_0,\ldots , p_{i-1}\}$ containing the first $i$ points of $P$.

The sequence $P$ is a \emph{greedy permutation} if for all $i$,
\[
  \dist(p_i, P_i) = \max_{p\in P}\dist(p, P_i).
\]
For a constant $\alpha > 1$, the sequence $P$ is an \emph{$\alpha$-approximate greedy permutation} if for all $i$,
\[
    \dist(p_i, P_i) \ge \alpha \max_{p\in P}\dist(p, P_i).
\]
The point $p_0 \in P$ is called the \emph{root} of the greedy permutation.

The \emph{predecessor mapping} $\pred:P \setminus \{p_0\}\to P$ maps each non-root point $p_i$ in $P$ to an approximate nearest neighbor in the prefix $P_i$.
The approximate nearest predecessor of a point need not be unique, however for the sake of construction, we assume we have chosen one.
The \emph{insertion distance} of a point $p \in P$, denoted $\e_p$, is defined as,
\[
    \e_ p = \begin{cases}
        \dist(p, T(p)) & p \ne p_0\\
        \infty & p = p_0.
    \end{cases}
\]
An \emph{$\alpha$-scaling} predecessor map is one where every point has an insertion distance that is at most $1/\alpha$-times the predecessor's insertion distance, i.e., for every $p \in P\setminus\{p_0\}$,
\[
    \e_p\le \frac{1}{\alpha}\e_{T(p)}.
\]
We will later show how a greedy permutation and a predecessor map can be used to build a hierarchical search data structure.

\paragraph*{Constructing Greedy Permutations}
We present a simple algorithm to construct $\alpha$-approximate greedy permutations.
The input is a point set in a metric space and an optional root point.
If no root is specified, then an arbitrary point can be chosen as the root.
The algorithm constructs the permutation iteratively while simultaneously constructing an $\alpha$-scaling predecessor map.
Each uninserted point keeps track of its $\alpha$-approximate nearest neighbor among the inserted points.
This point is called the \emph{parent} of the uninserted point.

We initialize the output permutation with only the root as an inserted point and all other points have the root as the parent.
In each iteration, the uninserted point $p$ maximizing the distance to its parent is inserted.
The predecessor map is updated to store the parent of $p$ as its predecessor.
For every uninserted point, update its parent to $p$ if $p$ is closer than its current parent by a factor of $\alpha$.
We refer to these parent updates as $\alpha$-lazy updates.
An example of this algorithm is shown in Figure~\ref{fig:greedy_perm}.

In each iteration, the inserted point is $\alpha$-approximately the farthest point and so the algorithm constructs an $\alpha$-approximate greedy permutation.
It is easy to see that if $\alpha=1$ then the output is an exact greedy permutation.
The $\alpha$-lazy updates guarantee that every point has a distance to its parent that is at most $1/\alpha$ times that parent's insertion distance.
So, it follows that the insertion distance of a point is at most $1/\alpha$-times that of its predecessor.
Therefore, the constructed predecessor map is $\alpha$-scaling.

\begin{figure}
    \begin{center}
    \begin{tikzpicture}[>={stealth[scale=3]}]
        \draw[thin] (-0.15,-0.15) rectangle (13.65,3.15);
        \draw[thin] (0,0) rectangle (3,3);
        
        \filldraw (0.1, 0.1) circle (0.02) node[right, xshift=5] {$a$};
        \filldraw (1.7, 1.5) circle (0.02) node[below] {$b$};
        \filldraw (2.5, 2.8) circle (0.02) node[right] {$c$};
        \filldraw (2.4, 1.7) circle (0.02) node[right] {$d$};

        \draw (0.1, 0.1) circle (0.05);

        \draw[gray, dashed, ->] (1.7, 1.5) -- (0.15, 0.15);
        \draw[dashed, ->] (2.5, 2.8) -- (0.1, 0.2);
        \draw[gray, dashed, ->] (2.4, 1.7) -- (0.2, 0.1);

        \def\xshft{3.5}
        \draw[thin] (0+\xshft,0) rectangle (3+\xshft,3);
        
        \filldraw (0.1+\xshft, 0.1) circle (0.02) node[right, xshift=5] {$a$};
        \filldraw (1.7+\xshft, 1.5) circle (0.02) node[below] {$b$};
        \filldraw (2.5+\xshft, 2.8) circle (0.02) node[right] {$c$};
        \filldraw (2.4+\xshft, 1.7) circle (0.02) node[right] {$d$};

        \draw (0.1+\xshft, 0.1) circle (0.05);
        \draw (2.5+\xshft, 2.8) circle (0.05);

        \draw[->] (2.45+\xshft, 2.75) -- (0.1+\xshft, 0.2);
        
        \draw[dashed, ->] (1.65+\xshft, 1.45) -- (0.15+\xshft, 0.15);
        \draw[gray, dashed, ->] (2.4+\xshft, 1.7) -- (2.5+\xshft, 2.75);

        \def\xshft{7}
        \draw[thin] (0+\xshft,0) rectangle (3+\xshft,3);

        \filldraw (0.1+\xshft, 0.1) circle (0.02) node[right, xshift=5] {$a$};
        \filldraw (1.7+\xshft, 1.5) circle (0.02) node[below] {$b$};
        \filldraw (2.5+\xshft, 2.8) circle (0.02) node[right] {$c$};
        \filldraw (2.4+\xshft, 1.7) circle (0.02) node[right] {$d$};

        \draw (0.1+\xshft, 0.1) circle (0.05);
        \draw (2.5+\xshft, 2.8) circle (0.05);
        \draw (1.7+\xshft, 1.5) circle (0.05);

        \draw[->] (2.45+\xshft, 2.75) -- (0.1+\xshft, 0.2);
        \draw[->] (1.65+\xshft, 1.45) -- (0.18+\xshft, 0.15);
        
        \draw[dashed, ->] (2.4+\xshft, 1.7) -- (2.5+\xshft, 2.75);
    
        \def\xshft{10.5}
        \draw[thin] (0+\xshft,0) rectangle (3+\xshft,3);

        \filldraw (0.1+\xshft, 0.1) circle (0.02) node[right, xshift=5] {$a$};
        \filldraw (1.7+\xshft, 1.5) circle (0.02) node[below] {$b$};
        \filldraw (2.5+\xshft, 2.8) circle (0.02) node[right] {$c$};
        \filldraw (2.4+\xshft, 1.7) circle (0.02) node[right] {$d$};

        \draw (0.1+\xshft, 0.1) circle (0.05);
        \draw (2.5+\xshft, 2.8) circle (0.05);
        \draw (1.7+\xshft, 1.5) circle (0.05);
        \draw (2.4+\xshft, 1.7) circle (0.05);

        \draw[->] (2.45+\xshft, 2.75) -- (0.1+\xshft, 0.2);
        \draw[->] (1.65+\xshft, 1.45) -- (0.18+\xshft, 0.15);
        \draw[->] (2.4+\xshft, 1.75) -- (2.5+\xshft, 2.7);

    \end{tikzpicture}

    \end{center}
    \caption{
    In this figure we compute the greedy permutation and predecessor mapping on a finite metric space $(\{a,b,c,d\}, \dist)$ with seed $a$ and $\alpha=2$.
    The dotted lines indicate parents.
    The darker dotted line highlights the point furthest from its parent.
    The solid lines indicate predecessors.
    Initially only $a$ is inserted and it is the parent of every uninserted point.
    When $c$ is inserted, it becomes the new parent of $d$ because $\alpha\dist(c,d) < \dist(a,d)$.
    The parent of $b$ is still $a$ even though $c$ is closer.
    Then, the insertion of $b$ does not change the parent of $d$ because $c$ is still an $\alpha$-approximate nearest neighbor.
    The completed permutation is $(a,c,b,d)$ and the predecessor map is $b \mapsto a, c \mapsto a, d \mapsto c$.
    }
    \label{fig:greedy_perm}
\end{figure}
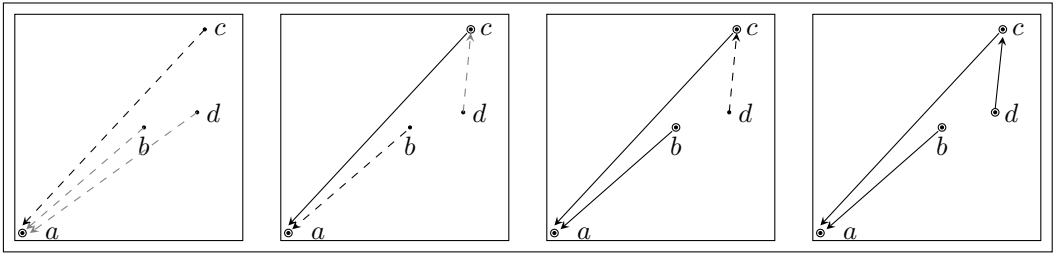

The algorithm presented above was analyzed by Gonzalez~\cite{Gonzalez85clustering} in 1985 and it runs in $O(n^2)$ time.
More efficient algorithms can construct greedy permutations in $O(n\log \spread)$ time for low-dimensional data~\cite{har-peled06fast,sheehy20greedypermutations}.
There also exists a randomized algorithm that can compute greedy permutations in $O(n\log n)$ time~\cite{har-peled06fast}.

\subsection{Greedy Trees}
A balltree~\cite{omohundro89five} is defined by recursively partitioning compact subsets of a metric space and representing the partitions in a binary tree.
For a ball tree on $A$, every node $x$ has a center $\nodectr{x} \in A$, and a radius $\noderad{x}$.
The set of all points contained in the leaves of $x$ is denoted by $\points(x)$.
The radius $\noderad{x}$ is such that $\points(x) \subseteq \ball(\nodectr{x}, \noderad{x})$.
Two nodes $x$ and $y$ are \emph{independent} if $\points(x) \cap \points(y) = \varnothing$.
We denote $\dist(\nodectr{x},\nodectr{y})$ as $\ctrdist(x,y)$.
%

A greedy tree~\cite{chubet23proximity} is a ball tree that can be built on $A$ using the greedy permutation on $A$ and a predecessor map on this permutation.
The root of the greedy tree is centered at the first point of the greedy permutation.
The rest of the tree is constructed incrementally.
At all times, there is a unique leaf centered at each of the previously inserted points.
For every point $p$ in the permutation, let $q$ denote its predecessor.
Create two nodes, one centered at $p$ and one centered at $q$.
Attach these nodes as the right and left children respectively of the leaf centered at $q$ (see Figure~\ref{fig:greedytreeinsert}).

The radius of a node is the distance from the center to the farthest leaf in its subtree.
Thus, leaves have a radius of zero.
The radii of all other nodes can be computed by traversing the subtrees.
The radius of a child is never greater than the radius of its parent.
If the greedy permutation is $\alpha$-approximate, then the resulting greedy tree is said to have a parameter $\alpha$.

The algorithm presented above has been analyzed in more detail in Chubet et al.~\cite{chubet23proximity}.
The following theorem appears as Theorem~{5.1} in Chubet et al.~\cite{chubet23proximity}.
\begin{restatable}{restate}{structuretheorem}
    \label{thm:structure}
    Let $G$ be a greedy tree with $\alpha > 1$.
    Then the following properties hold:
    \begin{enumerate}
        \item The radius of a node $x$ is bounded, $\noderad{x} \le \frac{\e_p}{\alpha-1}$.
        \item Let $X$ be a set of pairwise independent nodes from $G$.
        The centers of $X$ are $\frac{(\alpha-1)r}{\alpha}$-packed, where $r$ is the minimum radius of any parent of a node in $X$.
        \item The height of $G$ is $2^{O(d)}\log\spread$.
    \end{enumerate}
\end{restatable}
Theorem~\ref{thm:structure} allows us to bound the running time of the greedy tree construction algorithm.
Constructing the tree topology takes $O(n)$ time and computing radii takes $O(n\log \spread)$ time.
Therefore, given a greedy permutation and the approximate nearest predecessors, the corresponding greedy tree is computed in $O(n\log \spread)$ time.


\tikzset{
  treenode/.style = {align=center, inner sep=0pt, text centered,
    font=\sffamily},
  arn_n/.style = {treenode, circle, white, font=\sffamily\bfseries, draw=black,
    fill=blue!50, text width=1.25em},
  arn_r/.style = {treenode, circle, black, draw=blue!50,
    text width=1.25em, very thick}
}

\begin{figure*}
    \begin{center}
        \begin{tikzpicture}[scale=0.7, ->,level/.style={sibling distance = 3cm/#1, level distance = 1cm}]
                \node[label=$a$] at (2,7) {};
                \node[label=$b$] at (3,7) {};
                \node[label=$c$] (5) at (4,7) {};
                \node[label=$d$] (5) at (5,7) {};
                \node[label=$e$] (5) at (6,7) {};
                \node[label=$f$] (5) at (7,7) {};
            
                \draw[<-, black!100, thick] (2,7) to [bend right=80] (3,7);
                \draw[<-, black!100, thick] (3,8) to [bend left=80] (4,8);
                \draw[<-, black!100, thick] (2,8) to [bend right=-70] (5,8);
                \draw[<-, black!100, thick] (4,7) to [bend right=80] (6,7);
                \draw[<-, black!100, thick] (3.2,7) to [bend left=-80] (7,7);
        \end{tikzpicture}

        \begin{tikzpicture}
            \draw (0,0) -- (14,0);
        \end{tikzpicture}

        \vspace{0.5cm}
        
        \begin{tikzpicture}[scale=0.7, ->,level/.style={sibling distance = 3cm/#1, level distance = 1cm}]
            \begin{scope}[xshift=-3.5cm]
                \node [arn_n] {$u_0$} node[below, yshift=-8]{$a$};
            \end{scope}
            \begin{scope}[xshift=-1cm, level/.style={sibling distance = 1cm/#1,
                level distance = 1.25cm}]
                \node [arn_r] {$u_0$}
                child{ node [arn_n] {$u_1$} node[below, yshift=-8]{$a$}}
                child{node [arn_n] {$v_0$} node[below, yshift=-5]{$b$}};
            \end{scope}
            \begin{scope}[xshift=2.5cm, level/.style={sibling distance = 2cm/#1,
                level distance = 1cm}]
                \node [arn_r] {$u_0$}
                    child{ node [arn_r] {$u_1$}}
                    child{node [arn_r] {$v_0$}
                        child{ node [arn_n] {$v_1$} node[below, yshift=-5]{$b$}}
                        child{ node [arn_n] {$w_0$} node[below, yshift=-8]{$c$}}
                    };
            \end{scope}
            \begin{scope}[xshift=7.5cm, level/.style={sibling distance = 2cm/#1,
                level distance = 1cm}]
            \node [arn_r] {$u_0$}
                child{ node [arn_r] {$u_1$}
                    child{ node [arn_n] {$u_2$} node[below, yshift=-8]{$a$}}
                    child{ node [arn_n] {$x$} node[below, yshift=-5]{$d$}}
                }
                child{node [arn_r] {$v_0$}
                    child{ node [arn_r] {$v_1$}}
                    child{ node [arn_r] {$w_0$}}
                };
            \end{scope}
            \begin{scope}[xshift=13cm, level/.style={sibling distance = 2.5cm/#1,
                level distance = 1cm}]
            \node [arn_r] {$u_0$}
                child{ node [arn_r] {$u_1$}
                    child{ node [arn_r] {$u_2$}}
                    child{ node [arn_r] {$x$}}
                }
                child{node [arn_r] {$v_0$}
                    child{ node [arn_r] {$v_1$}}
                    child{ node [arn_r] {$w_0$}
                        child{ node [arn_n] {$w_1$} node[below, yshift=-8]{$c$}}
                        child{ node [arn_n] {$y$} node[below, yshift=-8]{$e$}}
                    }
                };
            \end{scope}
        \end{tikzpicture}

        \begin{tikzpicture}
            \draw (0,0) -- (14,0);
        \end{tikzpicture}

        \vspace{0.5cm}
        
        \begin{tikzpicture}[scale=0.7, ->,level/.style={sibling distance = 3cm/#1, level distance = 1cm}]
            \begin{scope}[yshift=3.75cm, xshift=3cm, level/.style={sibling distance = 5cm/#1,
                level distance = 1.5cm}]
            \node [arn_r] {$u_0$}
                child{ node [arn_r] {$u_1$}
                    child{ node [arn_r] {$u_2$} node[below, yshift=-8]{$a$}}
                    child{ node [arn_r] {$x$} node[below, yshift=-5]{$d$}}
                    node[below, yshift=-8]{$a$}}
                child{node [arn_r] {$v_0$}
                    child{ node [arn_r] {$v_1$}
                        child{ node [arn_r] {$v_2$} node[below, yshift=-5]{$b$}}
                        child{ node [arn_r] {$z$} node[below, yshift=-5]{$f$}}
                        node[below, yshift=-5]{$b$}
                    }
                    child{ node [arn_r] {$w_0$}
                        child{ node [arn_r] {$w_1$} node[below, yshift=-8]{$c$}}
                        child{ node [arn_r] {$y$} node[below, yshift=-8]{$e$}}
                        node[below, yshift=-8]{$c$}
                    }
                    node[below, yshift=-5]{$b$}
                } node[below, yshift=-8]{$a$};
        \end{scope}
        \end{tikzpicture}
    \end{center}
    \caption{
        In this figure, a ball tree is computed for a given permutation and predecessor pairing.
        The permutation $(a,b,c,d,e,f)$ is depicted with arrows representing a predecessor mapping at the top.
        The tree is constructed incrementally in the middle.
        Each new point creates two new nodes.
        The centers of newly inserted nodes are show below the nodes.
        The completed tree is shown at the bottom.
        The center of each node is depicted below it.
        There may be multiple nodes with the same center.
    }
    \label{fig:greedytreeinsert}
\end{figure*}
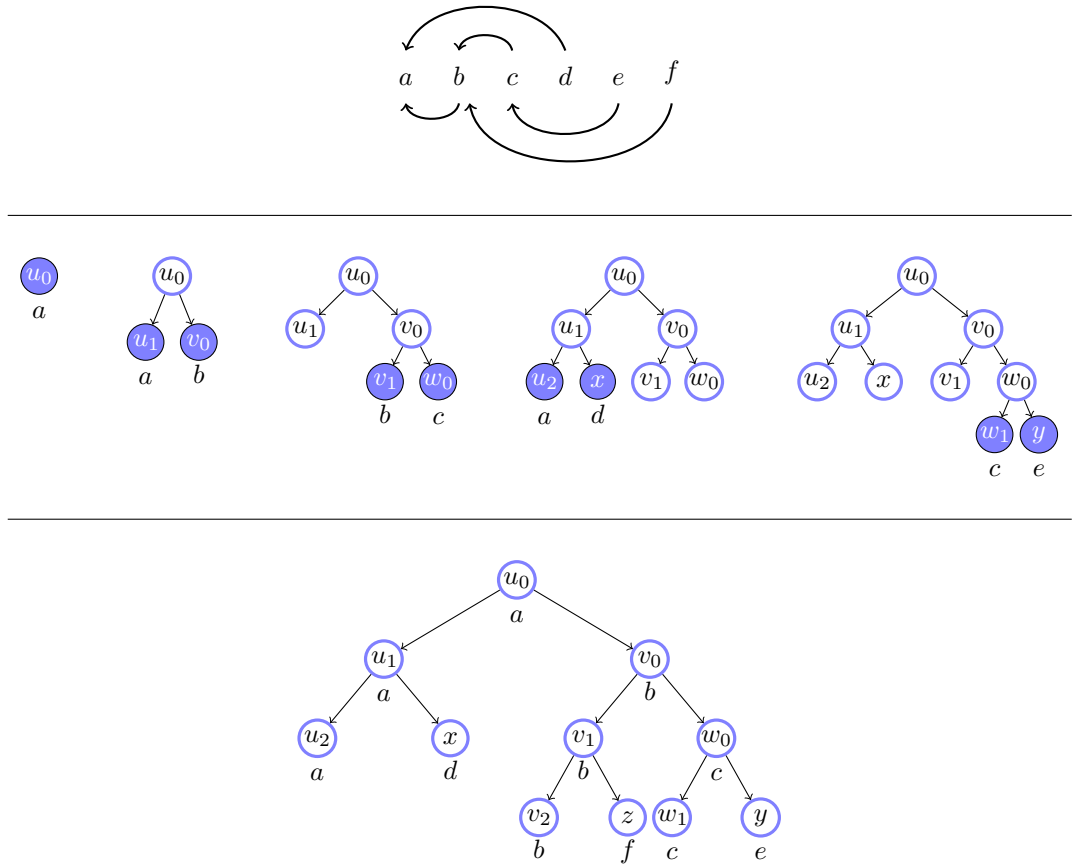


Once the tree $G$ is constructed for set $A$, it is traversed using a max heap that stores the nodes using the radius as the key.
The heap $H$ is initialized with the root node.
The traversal continues while there is a node in $H$ with non-zero radius.
In each iteration, the node with largest radius is removed from $H$ and its children are added to $H$.
We call this the \emph{radius-order traversal} of $G$.
See Figure~\ref{fig:greedy_tree_traverse} for an example of this traversal.
The following properties hold at every iteration of the radius-order traversal of $G$:
\begin{enumerate}
    \item \textbf{Covering}: For every point $a \in A$, there exists a unique node $x \in H$ such that $a \in \points(x)$.
    \item \textbf{Packing}:
    For every node in $H$, the radius of its parent is at least the radius of the node at the top of $H$.
    The nodes in $H$ are pairwise independent.
    So, by Theorem~\ref{thm:structure}, the nodes in $H$ are packed.
\end{enumerate}

\tikzset{
  treenode/.style = {align=center, inner sep=0pt, text centered,
    font=\sffamily},
  arn_n/.style = {treenode, circle, white, font=\sffamily\bfseries, draw=black,
    fill=blue!50, text width=1.25em},
  arn_r/.style = {treenode, circle, black, draw=blue!50,
    text width=1.25em, very thick}
}

\begin{figure}
    \begin{center}    
        \begin{tikzpicture}[>={stealth[scale=3]}, scale=0.8]
            \draw[thin] (0, 0) rectangle (5, 5);
            \filldraw[blue!30] (2.3, 2.5) circle (1.97);

            \draw (2.3, 2.5) circle (0.05);

            \draw[gray, ->] (3.6, 1) -- (2.35, 2.45);
            \draw[gray, ->] (3.9, 2) -- (3.63, 1.03);
            \draw[gray, ->] (1.8, 3.5) -- (2.25, 2.55);
            \draw[gray, ->] (3.4, 2.25) -- (3.85, 2);
            \draw[gray, ->] (2.9, 1.3) -- (3.55, 1);

            \filldraw (2.3, 2.5) circle (0.02) node[left]{$a$};
            \filldraw (3.6, 1) circle (0.02) node[right]{$b$};
            \filldraw (3.9, 2) circle (0.02) node[right]{$c$};
            \filldraw (1.8, 3.5) circle (0.02) node[right]{$d$};
            \filldraw (3.4, 2.25) circle (0.02) node[left]{$e$};
            \filldraw (2.9, 1.3) circle (0.02) node[left]{$f$};
        \end{tikzpicture}
        \hspace{1cm}
        \begin{tikzpicture}[scale=0.7, ->,level/.style={sibling distance = 3cm/#1, level distance = 1cm}]
            \begin{scope}[yshift=3.75cm, xshift=3cm, level/.style={sibling distance = 5cm/#1,
              level distance = 1.5cm}]
            \node [arn_n] {$u_0$}
                child{ node [arn_r] {$u_1$}
                    child{ node [arn_r] {$u_2$} node[below, yshift=-8]{$a$}}
                    child{ node [arn_r] {$x$} node[below, yshift=-5]{$d$}}
                    node[below, yshift=-8]{$a$}
                }
                child{node [arn_r] {$v_0$}
                    child{ node [arn_r] {$v_1$}
                        child{ node [arn_r] {$v_2$} node[below, yshift=-5]{$b$}}
                        child{ node [arn_r] {$z$} node[below, yshift=-5]{$f$}}
                    node[below, yshift=-8]{$b$}
                    }
                    child{ node [arn_r] {$w_0$}
                        child{ node [arn_r] {$w_1$} node[below, yshift=-8]{$c$}}
                        child{ node [arn_r] {$y$} node[below, yshift=-8]{$e$}}
                    node[below, yshift=-8]{$c$}
                    }
                   node[below, yshift=-8]{$b$}
                }
            node[below, yshift=-8]{$a$};
            \end{scope}
        \end{tikzpicture}

    \begin{tikzpicture}
        \draw (0,0) -- (14,0);
    \end{tikzpicture}

    \vspace{0.4cm}

        \begin{tikzpicture}[>={stealth[scale=3]}, scale=0.8]
            \filldraw[blue!30] (2.3, 2.5) circle (1.12);
            \filldraw[blue!30] (4.6, 0) arc (-45 : 225 : 1.33);

            \draw[gray, ->] (3.9, 2) -- (3.63, 1.03);
            \draw[gray, ->] (1.8, 3.5) -- (2.25, 2.55);
            \draw[gray, ->] (3.4, 2.25) -- (3.85, 2);
            \draw[gray, ->] (2.9, 1.3) -- (3.55, 1);

            \filldraw (2.3, 2.5) circle (0.02) node[left]{$a$};
            \filldraw (3.6, 1) circle (0.02) node[right]{$b$};
            \filldraw (3.9, 2) circle (0.02) node[right]{$c$};
            \filldraw (1.8, 3.5) circle (0.02) node[right]{$d$};
            \filldraw (3.4, 2.25) circle (0.02) node[left]{$e$};
            \filldraw (2.9, 1.3) circle (0.02) node[left]{$f$};

            \draw (2.3, 2.5) circle (0.05);
            \draw (3.6, 1) circle (0.05);

            \draw[thin] (0, 0) rectangle (5, 5);
        \end{tikzpicture}
        \hspace{1cm}
        \begin{tikzpicture}[scale=0.7, ->,level/.style={sibling distance = 3cm/#1, level distance = 1cm}]
            \begin{scope}[yshift=3.75cm, xshift=3cm, level/.style={sibling distance = 5cm/#1,
              level distance = 1.5cm}]
            \node [arn_r] {$u_0$}
                child{ node [arn_n] {$u_1$}
                    child{ node [arn_r] {$u_2$}}
                    child{ node [arn_r] {$x$}}
                }
                child{node [arn_n] {$v_0$}
                    child{ node [arn_r] {$v_1$}
                        child{ node [arn_r] {$v_2$}}
                        child{ node [arn_r] {$z$}}
                    }
                    child{ node [arn_r] {$w_0$}
                        child{ node [arn_r] {$w_1$}}
                        child{ node [arn_r] {$y$}}
                    }
                };
            \end{scope}
        \end{tikzpicture}

        \begin{tikzpicture}
            \draw (0,0) -- (14,0);
        \end{tikzpicture}

        \vspace{0.4cm}
    
        \begin{tikzpicture}[>={stealth[scale=3]}, scale=0.8]
            \draw[thin] (0, 0) rectangle (5, 5);

            \filldraw[blue!30] (2.3, 2.5) circle (1.12);
            \filldraw[blue!30] (3.6, 1) circle (0.75);
            \filldraw[blue!30] (3.9, 2) circle (0.56);

            \draw[gray, ->] (1.8, 3.5) -- (2.25, 2.55);
            \draw[gray, ->] (3.4, 2.25) -- (3.85, 2);
            \draw[gray, ->] (2.9, 1.3) -- (3.55, 1);

            \filldraw (2.3, 2.5) circle (0.02) node[left]{$a$};
            \filldraw (3.6, 1) circle (0.02) node[right]{$b$};
            \filldraw (3.9, 2) circle (0.02) node[right]{$c$};
            \filldraw (1.8, 3.5) circle (0.02) node[right]{$d$};
            \filldraw (3.4, 2.25) circle (0.02) node[left]{$e$};
            \filldraw (2.9, 1.3) circle (0.02) node[left]{$f$};

            \draw (2.3, 2.5) circle (0.05);
            \draw (3.6, 1) circle (0.05);
            \draw (3.9, 2) circle (0.05);

        \end{tikzpicture}
        \hspace{1cm}
        \begin{tikzpicture}[scale=0.7, ->,level/.style={sibling distance = 3cm/#1, level distance = 1cm}]
            \begin{scope}[yshift=3.75cm, xshift=3cm, level/.style={sibling distance = 5cm/#1,
                level distance = 1.5cm}]
              \node [arn_r] {$u_0$}
                  child{ node [arn_n] {$u_1$}
                      child{ node [arn_r] {$u_2$}}
                      child{ node [arn_r] {$x$}}
                  }
                  child{node [arn_r] {$v_0$}
                      child{ node [arn_n] {$v_1$}
                          child{ node [arn_r] {$v_2$}}
                          child{ node [arn_r] {$z$}}
                      }
                      child{ node [arn_n] {$w_0$}
                          child{ node [arn_r] {$w_1$}}
                          child{ node [arn_r] {$y$}}
                      }
                  };
              \end{scope}
        \end{tikzpicture}

        \begin{tikzpicture}
            \draw (0,0) -- (14,0);
        \end{tikzpicture}
    
        \vspace{0.4cm}
    
        \begin{tikzpicture}[>={stealth[scale=3]}, scale=0.8]
            \draw[thin] (0, 0) rectangle (5, 5);

            \filldraw[blue!30] (3.6, 1) circle (0.75);
            \filldraw[blue!30] (3.9, 2) circle (0.56);

            \draw[gray, ->] (3.4, 2.25) -- (3.85, 2);
            \draw[gray, ->] (2.9, 1.3) -- (3.55, 1);

            \filldraw (2.3, 2.5) circle (0.02) node[left]{$a$};
            \filldraw (3.6, 1) circle (0.02) node[right]{$b$};
            \filldraw (3.9, 2) circle (0.02) node[right]{$c$};
            \filldraw (1.8, 3.5) circle (0.02) node[right]{$d$};
            \filldraw (3.4, 2.25) circle (0.02) node[left]{$e$};
            \filldraw (2.9, 1.3) circle (0.02) node[left]{$f$};

            \draw (2.3, 2.5) circle (0.05);
            \draw (3.6, 1) circle (0.05);
            \draw (3.9, 2) circle (0.05);
            \draw (1.8, 3.5) circle (0.05);
        \end{tikzpicture}
        \hspace{1cm}
        \begin{tikzpicture}[scale=0.7, ->,level/.style={sibling distance = 3cm/#1, level distance = 1cm}]
            \begin{scope}[yshift=3.75cm, xshift=3cm, level/.style={sibling distance = 5cm/#1,
                level distance = 1.5cm}]
              \node [arn_r] {$u_0$}
                  child{ node [arn_r] {$u_1$}
                      child{ node [arn_n] {$u_2$}}
                      child{ node [arn_n] {$x$}}
                  }
                  child{node [arn_r] {$v_0$}
                      child{ node [arn_n] {$v_1$}
                          child{ node [arn_r] {$v_2$}}
                          child{ node [arn_r] {$z$}}
                      }
                      child{ node [arn_n] {$w_0$}
                          child{ node [arn_r] {$w_1$}}
                          child{ node [arn_r] {$y$}}
                      }
                  };
              \end{scope}
        \end{tikzpicture}
    \end{center}
    \caption{
        This figure shows a greedy tree and the first four iterations of its radius-order traversal.
        The center of every node is shown below the node in the first figure.
        In this tree $\nodectr{u_0} = \nodectr{u_1} = \nodectr{u_2}$ but $\noderad{u_0} \ne \noderad{u_1} \ne \noderad{u_2}$.
        Also, $\points(u_0) = \{a,b,c,d,e,f\}$ while $\points(u_1) = \{a,d\}$ and $\points(v) = \{b,c,e,f\}$.
        The heap is initialized with $u_0$.
        In each iteration, the node with the largest radius is replaced by its children.
        The complete order of traversal is $u_0, v_0, u_1, v_1, w_0$.
    }
    \label{fig:greedy_tree_traverse}
\end{figure}
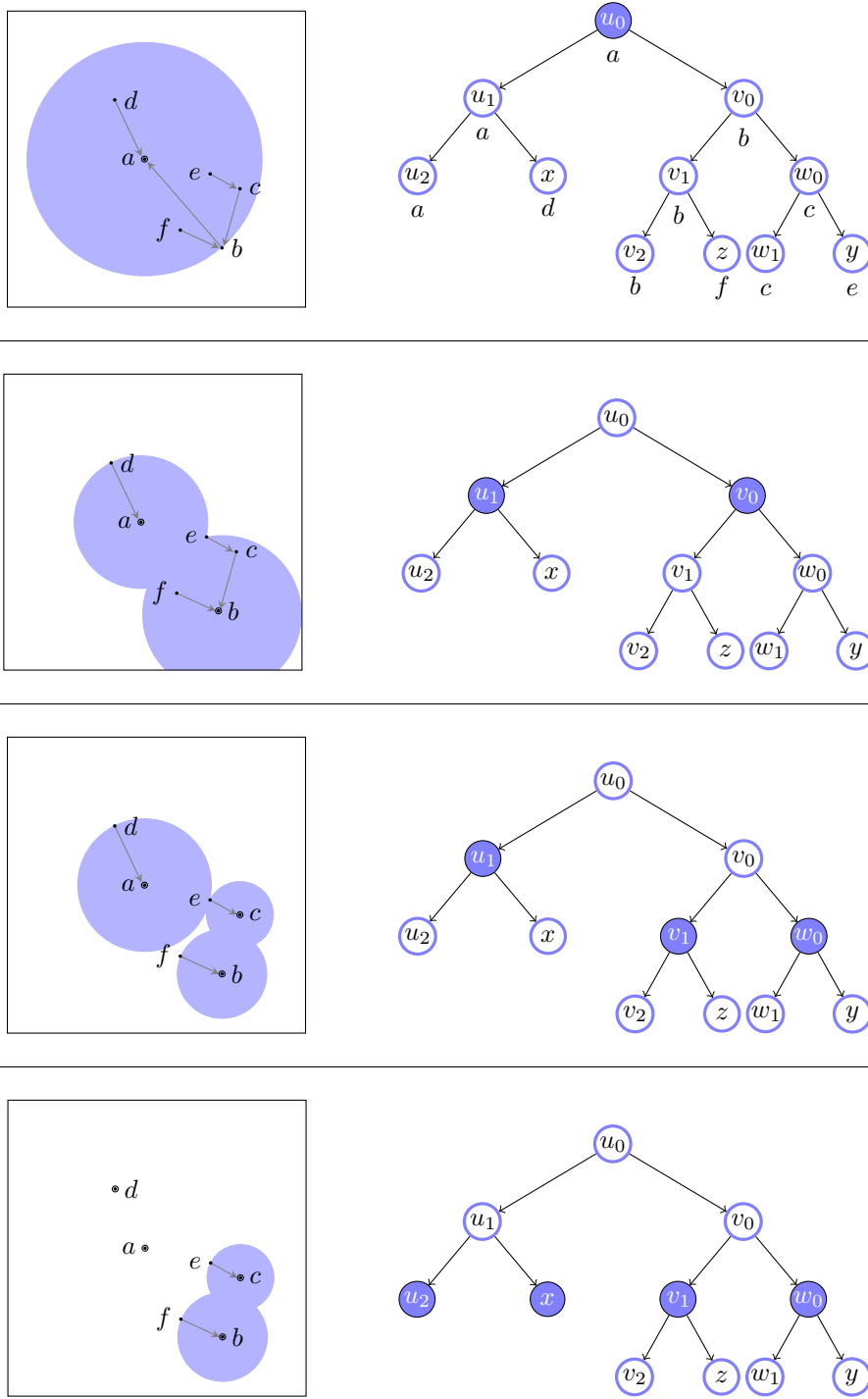



    \section{The Directed Hausdorff Distance}\label{sec:hd-algo}
In this section, we describe \algoHD, our directed Hausdorff distance approximation algorithm.
The main input to \algoHD is a pair of greedy trees $G_A$ and $G_B$ built on sets $A$ and $B$.
A simultaneous radius-order traversal of two greedy trees using a single heap can be done as described in the previous section.
The result is similar to a dual-tree traversal~\cite{curtin13tree,curtin15plug,gray00n-body,ram09linear}.
This traversal can be done beforehand and the output be stored in a list.
We assume we have access to the roots of the two greedy trees and the remaining nodes are stored in a single list sorted by non-increasing radii.
The preprocessing time does not exceed the tree construction time.
In Section~\ref{sec:mds} we give an application where this approach is useful.

The input also includes a parameter $\e> 0$ that determines the approximation factor.
The output is the $(1+\e)$-approximate directed Hausdorff distance from $A$ to $B$.

\subsection{The Setup}
The main data structure used is a bipartite graph $N$.
The two parts of $N$ are $N_A \subseteq G_A$ and $N_B \subseteq G_B$.
The set of neighbors of node $x$ in $N$ is denoted as $N(x)$.
This graph is called the \emph{viability graph} and it satisfies the following invariants:
\begin{itemize}
\item \textbf{Covering Invariant}: For every point in $A$ (and respectively, $B$), there exists a unique node in $N_A$ (respectively, $N_B$) that contains the point as a leaf.
\item \textbf{Edge Invariant}: For every point $a \in A$, if $\dist(a,B) = \dist(a,b)$ then there is an edge in $N$ between the nodes containing $a$ and $b$ as leaves.
\end{itemize}

In addition to the graph, \algoHD stores a \emph{local lower bound} on $\dist(a,B)$ for each point $a\in A$ that has been added to $N$.
Let node $x \in N_A$.
The local lower bound of point $\nodectr{x}$, denoted $\lb(x)$, is defined as,
\[
    \lb(x) := \max\left\{\min_{y \in N(x)}\ctrdist(x,y) - \noderad{y}, 0 \right\}.
\]
The largest of these lower bounds is stored as the \emph{global lower bound} $L$.
The global lower bound serves as an estimate of the directed Hausdorff distance.

\begin{figure}[h]
\begin{center}
\begin{tikzpicture}[thin,scale=0.8]

\def\green{black!50!green}
\draw (-0.5,-0.5) rectangle (7,5.5);
\foreach \x/\y/\r/\i in {4.25/0.5/0.5/2, 5.75/2.5/1/3}{
  \filldraw[gray!50] (\x,\y) circle (\r);
  \filldraw (\x,\y) circle (0.02) node[above,black] {$y_\i$};
};
\foreach \x/\y/\r/\i in {2.1/1/1.25/0, 2.5/3.5/0.5/1, 4.5/4.5/0.8/4}{
  \filldraw[blue!30] (\x,\y) circle (\r);
  \filldraw (\x,\y) circle (0.02) node[above,black] {$y_\i$};
};
\filldraw[\green!50] (0.75,3.5) circle (1);
\filldraw (0.75,3.5) circle (0.02) node[black, above, yshift=2] {$x$};
\draw[|-|] (0.75,3.5) -- ++(-1,0);
\draw[|-|] (0.75,3.5) -- (2.5,3.5);
\draw[|-|] (0.75,2.5) -- ++(0,-1.75);
\draw[|-|] (0.75,0.75) -- ++(0,-1);
\draw[dashed, shift=(-109:3.75)] (0.75,3.5) arc (-109:32:3.75);
\end{tikzpicture}
\end{center}
\caption{
  The the nodes $y_2$ and $y_3$ are too far away to contain the nearest neighbor of any point in $x$, so edges $(x,y_2)$ and $(x,y_3)$ can be pruned from the viability graph.
  The pruning condition respects the neighbor invariant and does not prune edge $(x,y_4)$.
}
\label{fig:pruning}
\end{figure}
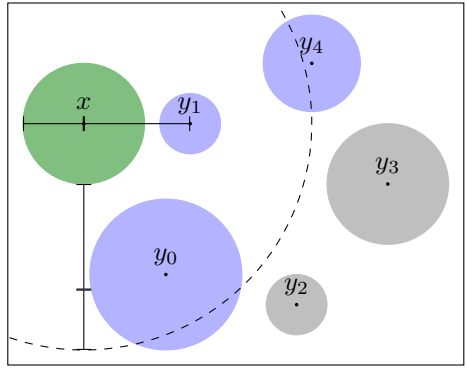



\subsection{The Algorithm}

Now, we describe \algoHD.
Let $x_0$ and $y_0$ be the root nodes of $G_A$ and $G_B$ respectively.
Initialize $N$ to contain $x_0$ and $y_0$ as vertices connected by an edge.
Initialize $\lb(x_0)$ and $L$.
The main loop iterates over the input list of greedy tree nodes.

Let $z$ be the next node in the list.
If $\noderad{z} \le \frac{\e L}{2}$, then it is safe to ignore the remaining nodes and terminate the loop.
We call this inequality the \emph{stopping condition}.
Return $L$ as $(1+\e)$-approximate directed Hausdorff distance between $A$ and $B$.
Else, let $z_l$ and $z_r$ be the left and right children of $z$ respectively.
Add $z_l$ and $z_r$ as vertices in $N$.
Connect the children to the neighbors of $z$.
Remove $z$ from $N$.
If $z \in N_A$, let $S:=\{z_l, z_r\}$, else let $S:=N(z_l)$.

For each $x \in S$ iterate over the neighbors of $x$.
An edge $(x, y)$ can be removed if there exists $y' \in N(x)$ such that
\[
    \ctrdist(x,y) - \noderad{y} > \ctrdist(x, y') + 2\cdot\noderad{x}.
\]
Remove all such edges incident to $x$ (see Figure~\ref{fig:iteration}).
We refer to this step as \emph{pruning} node $x$.
Finally, update $\lb(x)$.
A lower bound update is shown in Figure~\ref{fig:locallowerbound}.
Repeat the main loop until the stopping condition is satisfied.


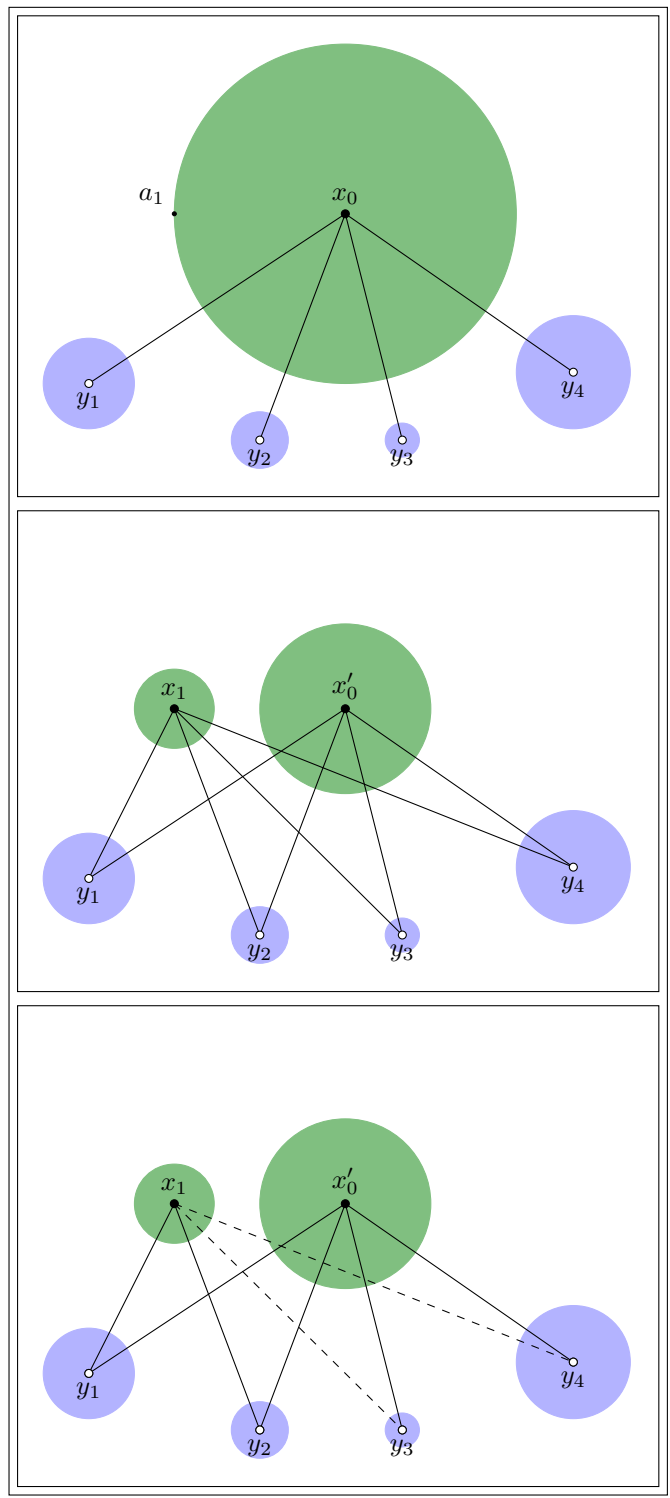
\begin{figure}
    \centering
    \begin{tikzpicture}[scale=0.75]
        \draw (-2.9, 5.65) rectangle (8.65, -20.65);
        \draw (-2.75, -3) rectangle (8.5,5.5);
        \def\r{0.07}
        \def\p{(3,2)}
        \def\q{(0,2)}
        \def\points{(-1.5,-1)/0.8/$y_1$, (1.5,-2)/0.5/$y_2$, (4,-2)/0.3/$y_3$, (7, -0.8)/1/$y_4$}
        \def\green{black!50!green!50}
        \def\blue{blue!30}
        \def\gray{gray!50}
        \def\darr{stealth-stealth}
        
        \filldraw[\green] \p circle (3);
        \filldraw \p circle (\r) node[above] {$x_0$};

        \filldraw \q circle (\r/2) node[above left] {$a_1$};

        \foreach \a/\rad/\label in \points
        {
            \filldraw[\blue] \a circle (\rad);
            \draw \p -- \a;
            \draw[fill=white] \a circle (\r) node[below] {\label};
        }

        \def\yshft{-8.75}
        \draw (-2.75, -3+\yshft) rectangle (8.5,5.5+\yshft);
        \def\r{0.07}
        \def\p{(3,2+\yshft)}
        \def\q{(0,2+\yshft)}
        \def\s{(0.5,1+\yshft)}
        \def\points{(-1.5,-1+\yshft)/0.8/$y_1$, (1.5,-2+\yshft)/0.5/$y_2$, (4,-2+\yshft)/0.3/$y_3$, (7, -0.8+\yshft)/1/$y_4$}
        \def\green{black!50!green!50}
        \def\blue{blue!30}
        \def\gray{gray!50}
        \def\darr{stealth-stealth}
        
        \filldraw[white] \p circle (3);
        \filldraw[\green] \p circle (1.5);
        \filldraw \p circle (\r) node[above] {$x_0'$};

        \filldraw[\green] \q circle (0.7);
        \filldraw \q circle (\r) node[above] {$x_1$};

        \foreach \a\rad\label in \points
        {
            \filldraw[\blue] \a circle (\rad);
            \draw \p -- \a;
            \draw \q -- \a;
            \draw[fill=white] \a circle (\r) node[below] {\label};
        }

        \def\yshft{-17.5}
        \draw (-2.75, -3+\yshft) rectangle (8.5,5.5+\yshft);
        \def\r{0.07}
        \def\p{(3,2+\yshft)}
        \def\q{(0,2+\yshft)}
        \def\s{(0.5,1+\yshft)}
        \def\points{(-1.5,-1+\yshft)/0.8/$y_1$, (1.5,-2+\yshft)/0.5/$y_2$, (4,-2+\yshft)/0.3/$y_3$, (7, -0.8+\yshft)/1/$y_4$}
        \def\green{black!50!green!50}
        \def\blue{blue!30}
        \def\gray{gray!50}
        \def\darr{stealth-stealth}
        
        \filldraw[white] \p circle (3);
        \filldraw[\green] \p circle (1.5);
        \filldraw \p circle (\r) node[above] {$x_0'$};

        \filldraw[\green] \q circle (0.7);
        \filldraw \q circle (\r) node[above] {$x_1$};

        \foreach \a\rad\label in \points
        {
            \filldraw[\blue] \a circle (\rad);
            \draw \p -- \a;
            \draw[fill=white] \a circle (\r) node [below] {\label};
        }
        \draw \q -- (-1.5,-1+\yshft);
        \draw[fill=white] (-1.5, -1+\yshft) circle (\r);
        \draw \q -- (1.5,-2+\yshft);
        \draw[fill=white] (1.5,-2+\yshft) circle (\r);
        \draw[dashed] \q -- (4,-2+\yshft);
        \draw[fill=white] (4,-2+\yshft) circle (\r);
        \draw[dashed] \q -- (7, -0.8+\yshft);
        \draw[fill=white] (7, -0.8+\yshft) circle (\r);
        
    \end{tikzpicture}
    \caption{This figure illustrates pruning in an iteration of \algoHD.
    At the top, the node ${x_0}$ is connected to nodes ${y_1}, {y_2}, {y_3}$, and ${y_4}$.
    When ${x_0}$ is split and replaced by its children (middle image), edges are added between the new node ${x_1}$ and the neighbors of its parent.
    Then, at the bottom, the edges that are too long are pruned using the pruning condition.
    In this case, $({x_1}, {y_3})$ and $({x_1}, {y_4})$ are pruned.}
    \label{fig:iteration}
\end{figure}
 
\begin{figure}[h]
\begin{center}
\begin{tikzpicture}[thick,scale=0.9]
\draw[thin] (-3.25, -0.9) rectangle (12.15, 3.65);
\def\green{black!50!green}
\draw[thin] (-3.1,-0.75) rectangle (4,3.5);
\draw[\green] (1,1) edge node[black, below right,xshift=9,yshift=2] {\small$\lb(x)$}(-1.5,1);
\foreach \x/\y/\r/\i in {-1.5/1/1.5/0, 3/2.5/0.75/1, 2.7/0/0.5/2}{
  \draw[thick,dotted,\green] (1,1) -- (\x,\y);
  \filldraw[blue!30] (\x,\y) circle (\r);
  \filldraw (\x,\y) circle (0.02) node[above,black] {$y_\i$};
};
\draw[\green] (1,1) circle (0.02) node[black, right, xshift=1] {$x$};

\def\xshft{8}
\draw[thin] (-3.1+\xshft,-0.75) rectangle (4+\xshft,3.5);
\draw[\green] (1+\xshft,1) edge node[black, below right,xshift=-20,yshift=4] {\small$\lb(x)$}(2.7+\xshft,0);
\foreach \x/\y/\r/\i in {-1.5/1/0.5/$y_0'$, 3/2.5/0.75/$y_1$, 2.7/0/0.5/$y_2$, -1.25/2/0.4/$y_3$}{
  \draw[thick,dotted,\green] (1+\xshft,1) -- (\x+\xshft,\y);
  \filldraw[blue!30] (\x+\xshft,\y) circle (\r);
  \filldraw (\x+\xshft,\y) circle (0.02) node[above,black] {\i};
};
\draw[\green] (1+\xshft,1) circle (0.02) node[black, right, xshift=1] {$x$};
\end{tikzpicture}
\end{center}
\caption{
This figure depicts an update of $\lb(x)$ after replacing a node $y_0$ with its children.
}
\label{fig:locallowerbound}
\end{figure}
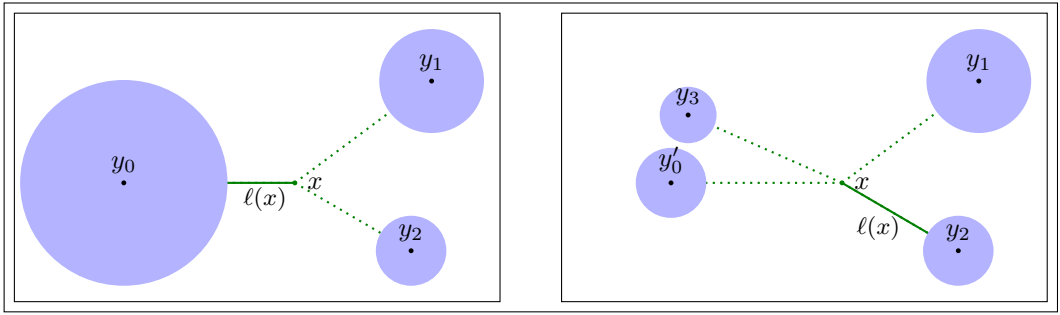


\begin{figure}[h]
\begin{center}
    \begin{tikzpicture}[scale=0.8]
\def\green{black!50!green!50}
\def\blue{blue!30}
\def\gray{gray!50}
\def\r{0.03}
\def\darr{stealth-stealth}

\draw (-1.5,-0.2) rectangle (7,5.2);

\filldraw[\gray] (-0.6,2) circle (0.8);
\filldraw[\green] (4,4) circle (1);
\filldraw[\blue] (6,4) circle (0.65);

\filldraw (-0.6,2) circle (\r) node[above left] {$y'$};
\filldraw (0,2.2) circle (\r) node[above left] {$b$};
\filldraw (4,4) circle (\r) node[above left]{$x$};
\filldraw (3.5,3.5) circle (\r) node[above left]{$a$};
\filldraw (6,4) circle (\r) node[above left]{$y$};
\draw[\darr] (3,4) -- (4,4);
\draw[\darr] (6,4) -- (4,4);
\draw[\darr] (4,3) -- (4,1);
\draw[\darr] (4,0) -- (4,1);
\draw[dashed,shift=(-197:4)] (4,4) arc (-197:-42:4);

\end{tikzpicture}
\end{center}
\caption{
    Let $x \in N_A$.
    For every point $a$ in $\points(x)$, we have $\dist(a,B) \le \ctrdist(x,y) + \noderad{x}$ by the triangle inequality.
    If edge $(x, y')$ has been pruned, then no point $b \in \points(y')$ can be the nearest neighbor of all $a \in \points(x)$, because $\dist(a,\nodectr{y}) \le \dist(a,b)$ for any such $b$.
}
\label{fig:nbrinvariant}
\end{figure}




\subsection{Analysis}

Figure~\ref{fig:nbrinvariant} illustrates how the pruning condition does not violate the Edge Invariant.
This is formalized in the following lemma:

\begin{lemma}\label{lem:pruning_correctness}
    For $a \in A$, let $b \in B$ such that $\dist(a,B) = \dist(a,b)$.
    Let $x \in N_A$ and $y \in N_B$ be such that $a \in \points(x)$ and $b \in \points(y)$ in some iteration of \algoHD.
    Then, there exists an edge between $x$ and $y$ in $N$.
\end{lemma}
\begin{proof}
Given $a\in A$, let $b \in B$ such that $\dist(a,b) = \dist(a,B)$.
By the Covering Invariant, there exists $x \in N_A$ and $y \in N_B$ such that $a \in \points(x)$ and $b\in \points(y)$.
Then, by the triangle inequality, $\dist(a,b) \ge \ctrdist(x,y) - \noderad{x} - \noderad{y}$.
Suppose for the sake of contradiction, the edge $(x,y)$ in $N$ was pruned in the current iteration.
Then $(x,y)$ must satisfy the pruning condition, i.e., for some $y' \in N_B$ we have $\ctrdist(x,y) - \noderad{y} > \ctrdist(x,y') + 2\cdot\noderad{x}$.
It follows that $\dist(a,b) > \ctrdist(x,y') + \noderad{x}$.
Additionally, $\ctrdist(x, y') \ge \dist(a, \nodectr{y'}) - \noderad{x}$ by the triangle inequality.
Thus, we have $\dist(a,b) > \dist(a, \nodectr{y'})$.
This is a contradiction as $b$ is a nearest neighbor of $a$.
Therefore, there exists an edge between $x$ and $y$.
\end{proof}

\begin{lemma}\label{lem:l_plus_2r}
Let $r$ be the radius of the node to be processed in an iteration of \algoHD and let $L$ be the global lower bound.
Then, $L \le \dist_h(A,B) \le L + 2r$.
Moreover, if $r \le (\frac{\e }{2})L$, then $L$ is a $(1 + \e)$-approximation of $\dist_h(A, B)$.
\end{lemma}
\begin{proof}
Let $A' := \{\nodectr{x} \mid x \in N_A\}$ and $B' := \{\nodectr{y} \mid y \in N_B\}$.
We first show that the distance from $A'$ to $B$ is at most $L+r$.
We know that $\dist_h(A', B) \le \dist_h(A', B')$ because for any $a \in A'$, we have $\dist(a,B) \le \dist(a,B')$.
Also, $\noderad{z} \le r$ for any vertex $z \in N$.
So,
\[
    \dist(a,B') \le \min_{y\in N_B} \{\dist(a,\nodectr{y}) - \noderad{y} + r\} \le \lb(x)+r,
\]
where $x \in N_A$ is the node with $\nodectr{x} = a$.
It follows that,
\[
    \dist_h(A',B) \le \max_{x\in N_A} \{\lb(x) + r\}= L + r.
\]
Furthermore, $\dist(a,B) \le \dist(a,A') + \dist_h(A', B)$, and we know that $\dist(a,A') \le r$.
    So, $\dist(a,B) \le L+2r$.
    Therefore $L \le \dist_h(A,B) \le L + 2r$.
    It follows that if $r \le \frac \e 2 L$ then $\dist_h(A,B) \le (1+\e)L$.
\end{proof}

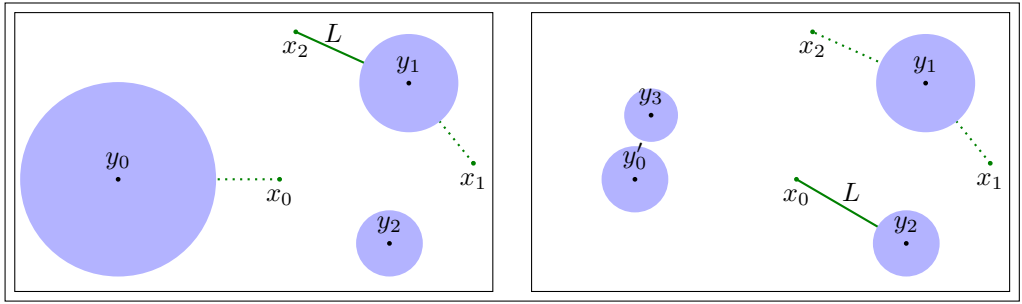
\begin{figure}
\begin{center}
\begin{tikzpicture}[thick,scale=0.85]
\draw[thin] (-3.25,-0.9) rectangle (12.45,3.75);
\draw[thin] (-3.1,-0.75) rectangle (4.3,3.6);
\def\green{black!50!green}
\draw[\green] (3,2.5) edge node[black,above left, yshift=2] {$L$} (1.25,3.3);
\draw[\green,dotted] (1,1) edge (-1.5,1);
\draw[\green,dotted] (4,1.25) edge (3,2.5);
\foreach \x/\y/\r/\i in {-1.5/1/1.5/0, 3/2.5/0.75/1, 2.7/0/0.5/2}{
  \filldraw[blue!30] (\x,\y) circle (\r);
  \filldraw (\x,\y) circle (0.02) node[above,black] {$y_\i$};
};
\foreach \x/\y/\i in {1/1/0, 4/1.25/1, 1.25/3.3/2}{
    \filldraw[\green] (\x,\y) circle (0.02) node[black, below] {$x_\i$};
}
\def\xshft{8}
\def\green{black!50!green}
\draw[thin] (-3.1+\xshft,-0.75) rectangle (4.3+\xshft,3.6);
\draw[\green] (2.7+\xshft,0) edge node[black,above, yshift=0] {$L$} (1+\xshft,1);
\draw[\green,dotted] (3+\xshft,2.5) edge (1.25+\xshft,3.3);
\draw[\green,dotted] (4+\xshft,1.25) edge (3+\xshft,2.5);
\foreach \x/\y/\r/\i in {-1.5/1/0.5/$y_0'$, 3/2.5/0.75/$y_1$, 2.7/0/0.5/$y_2$, -1.25/2/0.4/$y_3$}{
  \filldraw[blue!30] (\x+\xshft,\y) circle (\r);
  \filldraw (\x+\xshft,\y) circle (0.02) node[above,black] {\i};
};
\foreach \x/\y/\i in {1/1/0, 4/1.25/1, 1.25/3.3/2}{
    \filldraw[\green] (\x+\xshft,\y) circle (0.02) node[black, below] {$x_\i$};
}
\end{tikzpicture}
\end{center}
\caption{
This figure depicts an update of $L$ after replacing the node $y_0$ with its children. The directed Hausdorff distance, $\dist_h(A,B) \ge L$.
}
\label{fig:globallowerbound}
\end{figure}




At all times, the degree of every vertex in $N$ is bounded by a constant.
We will show that this invariant is guaranteed by the pruning algorithm, the early stopping condition, and a packing bound.
It is the critical fact in the following analysis of the \algoHD running time.

\begin{theorem}
    \label{thm:hd_correct}
    Given two greedy trees for sets $A$ and $B$ of total cardinality $n$, $\algoHD$ computes a $(1+\e)$-approximation of $\dist_h(A,B)$ in $\left (2+\frac{1}{\e}\right )^{O(d)}n$ time.
\end{theorem}
\begin{proof}
    Lemmas~\ref{lem:pruning_correctness} and~\ref{lem:l_plus_2r} imply the correctness of \algoHD.
    Now, we bound the running time.
    Let $G_A$ and $G_B$ be $\alpha$-approximate greedy trees.
    In order to bound the degrees of the viability graph $N$, we first establish that the points associated with the neighbors of a vertex in $N$ are packed.
    Let $r$ be the radius of the next node in the input list.
    By construction, any node $z \in N$ is the left or right child of a greedy tree node with radius at least $r$.
    Then by Theorem~\ref{thm:structure}, the centers of nodes in $N$ are $\frac{(\alpha-1)r}{\alpha}$-packed.
    The pruning condition implies that the distance from $\nodectr{z}$ to any of its neighbors is at most $L + 4r$.
    Thus, by Lemma~\ref{lem:std_packing},
    \[
        |N(z)|
            \le \Big(\frac{2\alpha(L+4r)}{(\alpha-1)r}\Big)^d
            \le \Big(\frac{16\alpha(1+\e)}{(\alpha-1)\e}\Big)^d,
    \]
    for all $r \ge (\frac{\e}{2})L$.
    Therefore, the number of edges incident to any given node in $N$ is $\left (2+\frac{1}{\e}\right )^{O(d)}$.
    So we spend $\left (2+\frac{1}{\e}\right )^{O(d)}$ time for each iteration of the algorithm.
    This gives a running time of $\left (2+\frac{1}{\e}\right )^{O(d)}n$.
\end{proof}



    \section{The Partial Directed Hausdorff Distance}\label{sec:kpartial}
Let $A$ and $B$ be subsets of a metric space $(X,\dist)$.
Let $A^{(k)}$ denote all subsets of $A$ with $k$ elements removed.
That is,
$A^{(k)} := \{S \subseteq A : |A\backslash S|=k\}$.
The \emph{$\kth$-partial} directed Hausdorff distance is
    $\khdist{k}(A,B) := \min_{S\in A^{(k)}} \dist_h(S,B)$.
Equivalently, $\khdist{k}(A,B)$ is the $(k+1)^\textrm{st}$ largest distance $\dist(a,B)$ over all $a\in A$.
In particular, $\dist_h^{(0)} = \dist_h$, as shown in Figure~\ref{fig:partialhd}.

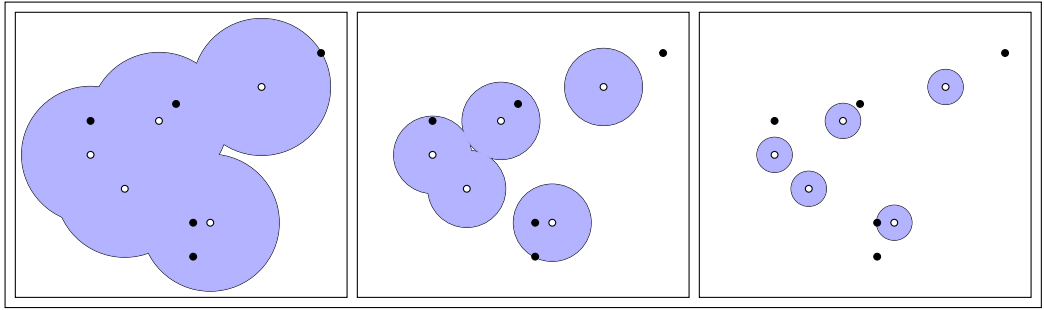
\begin{figure*}[h]
\begin{center}
\def\green{black!50!green!50}
\def\blue{blue!30}
\def\white{white!30}
\def\r{0.05}
\begin{tikzpicture}[scale=0.9]
\def\hd{1}
\draw (-1.25,-2.25) rectangle (13.9,2.25);
\draw (-1.1,-2.1) rectangle (3.75,2.1);
\foreach \x/\y in {0/0, 1/0.5, 0.5/-0.5, 1.75/-1, 2.5/1}{
  \draw[thick] (\x,\y) circle (\hd);
}
\foreach \x/\y in {0/0, 1/0.5, 0.5/-0.5, 1.75/-1, 2.5/1}{
  \filldraw[\blue] (\x,\y) circle (\hd);
}
\foreach \x/\y in {0/0, 1/0.5, 0.5/-0.5, 1.75/-1, 2.5/1}{
  \draw[fill=white] (\x,\y) circle (\r);
}
\foreach \x/\y in {0/0.5, 1.25/0.75, 1.5/-1.5, 1.5/-1, 3.37/1.5}{
  \draw[fill=black] (\x,\y) circle (\r);
}
\def\hd{0.56}
\def\xshft{5}
\draw (-1.1+\xshft,-2.1) rectangle (3.75+\xshft,2.1);
\foreach \x/\y in {0/0, 1/0.5, 0.5/-0.5, 1.75/-1, 2.5/1}{
  \draw[thick] (\x+\xshft,\y) circle (\hd);
}
\foreach \x/\y in {0/0, 1/0.5, 0.5/-0.5, 1.75/-1, 2.5/1}{
  \filldraw[\blue] (\x+\xshft,\y) circle (\hd);
}
\foreach \x/\y in {0/0, 1/0.5, 0.5/-0.5, 1.75/-1, 2.5/1}{
  \draw[fill=white] (\x+\xshft,\y) circle (\r);
}
\foreach \x/\y in {0/0.5, 1.25/0.75, 1.5/-1.5, 1.5/-1, 3.37/1.5}{
  \draw[fill=black] (\x+\xshft,\y) circle (\r);
}
\def\hd{0.25}
\def\xshft{10}
\draw (-1.1+\xshft,-2.1) rectangle (3.75+\xshft,2.1);
\foreach \x/\y in {0/0, 1/0.5, 0.5/-0.5, 1.75/-1, 2.5/1}{
  \draw[thick] (\x+\xshft,\y) circle (\hd);
}
\foreach \x/\y in {0/0, 1/0.5, 0.5/-0.5, 1.75/-1, 2.5/1}{
  \filldraw[\blue] (\x+\xshft,\y) circle (\hd);
}
\foreach \x/\y in {0/0, 1/0.5, 0.5/-0.5, 1.75/-1, 2.5/1}{
  \draw[fill=white] (\x+\xshft,\y) circle (\r);
}
\foreach \x/\y in {0/0.5, 1.25/0.75, 1.5/-1.5, 1.5/-1, 3.37/1.5}{
  \draw[fill=black] (\x+\xshft,\y) circle (\r);
}
\end{tikzpicture}
\end{center}
\caption{Depicted above are the directed Hausdorff distance $\dist_h(A,B)$ (left),
the first partial Hausdorff distance $\dist_h^{(1)}(A,B)$ (center),
and the $4^\text{th}$-partial Hausdorff distance $\dist_h^{(4)}(A,B)$ (right).}
\label{fig:partialhd}
\end{figure*}



The stopping condition of \algoHD is satisfied when it has discovered a node $x$ with center $a$ such that $\dist(a,B)$ is approximately the largest among all points in $A$.
Instead of terminating the loop at this point, if we continue running the algorithm ignoring $a$, the algorithm discovers the next point in $A$ approximately farthest to $B$ the next time the stopping condition is satisfied.
So a simple modification of \algoHD gives us an algorithm to approximate the $\kth$-partial directed Hausdorff distance.

The algorithm $\algokHD$ approximates $\khdist{k}(A,B)$ for all $k \le |A|$.
The input to \algokHD consists of two greedy trees and an approximation parameter preprocessed in the same way as \algoHD.
The output is a sequence $(\delta_0,\dots, \delta_{n-1})$ of distances such that $\delta_i \le \khdist{i}(A,B) \le (1+\e)\delta_i$.
The running time of the algorithm depends on a novel heap data structure.
We present the algorithm first and then discuss data structure choices before analyzing the running time.

\subsection{The Setup}
\algokHD maintains the same bipartite viability graph $N$ and local lower bounds as \algoHD.
Instead of storing just a global lower bound $L$, \algokHD stores all the nodes with centers in $A$ in a separate max heap $H$ with the local lower bounds as keys.
This heap is called the \emph{lower bound heap}.
Every local lower bound update also updates the priority of the node in $H$.
Additionally, there is a list to store the output.
    
\subsection{The Algorithm}\label{subsec:fractional_algo}
The initialization of \algokHD is the same as that of \algoHD.
The output list is initialized to be empty.

The main loop of this algorithm iterates over the input list in a manner similar to that of \algoHD.
The only change is to the stopping condition.
Instead of checking the stopping condition, check the following \emph{finishing condition} at the beginning of each iteration.
Let $r$ be the radius of the current node in the input list.
Let $x$ be the node at the top of $H$.
While $r \le \frac{\e}{2}\lb(x)$ \emph{finish} node $x$ as follows:
\begin{enumerate}
    \item Remove $x$ from $H$.
    \item Append $\lb(x)$ to the output list for each $a \in \points(x)$.
    \item Remove all edges incident to $x$ in $N$.
\end{enumerate}
Instead of terminating the loop when the finishing condition is satisfied, as in \algoHD, proceed and update the vertices in $N$ while pruning nodes to remove long edges.
The loop runs over the whole input list.

\paragraph*{Lower Bound Heap}\label{sec:lbh}

In \algokHD, the running time of each iteration is determined by the cost of updates to the lower bound heap.
Using a standard heap would require $O(\log n)$ time for basic operations.
Allowing for a small approximation, we can put the nodes in buckets, where the $m^{\text{th}}$ bucket of this heap contains nodes with radii in the interval $(\beta^m, \beta^{m+1}]$, for $1 < \beta \le 1 + \e$.
There are at most $O(\log_\beta \spread)$ buckets.
Using a standard heap on the buckets would lead to $O(\log\log_\beta \spread)$ time for basic operations.
This can be improved still further using a bucket queue~\cite{skiena98algorithms}, an array of the $O(\log_\beta \spread)$ buckets.
Most operations will take constant time except for remove max, which can require iterating through the buckets.
We will show that the buckets are removed in non-increasing order so the iteration visits each bucket at most once, incurring a total cost of $O(\log_\beta \spread)$.

A problem in using a bucket queue is that local lower bounds may increase over the course of the algorithm.
A na\"{i}ve approach is inefficient as it may end up revisiting buckets.
However, the following observation holds by the triangle inequality (see Figure~\ref{fig:inc_lower_bounds}).

\begin{figure}
    \begin{center}
        \begin{tikzpicture}
            \draw (-0.5, -0.65) rectangle (11.5, 1.15); 
            \foreach \i in {0,...,10}{
                \draw (0+1*\i,0) rectangle (1+1*\i,1);
            }
            \draw[thick, white] (0, 0) -- (0, 1);
            \draw[thick, white] (11, 0) -- (11, 1);

            \draw (0.5, 0.5) node {$\dots$};
            \draw (10.5, 0.5) node {$\dots$};
            
            \draw (0,0) node[below] {$\beta^{-\infty}$};
            \draw (8, 0) node[yshift=-1.5, below] {$\beta^m$};
            \draw (9, 0) node[below] {$\beta^{m+1}$};
            \draw (11,0) node[yshift=-1.5, below] {$\beta^{\infty}$};

            \draw (2.5, 0.5) node {$z$};
            \draw (3.5, 0.5) node {$u, v$};
            \draw (6.5, 0.5) node {$w$};
            \draw (8.5, 0.5) node {$x, y$};

            \draw[<-] (4.5, 0.35) arc (-60:-120:2);
        \end{tikzpicture}
    \end{center}
    \caption{
        This figure shows how a $\beta$-bucket queue implements an approximate lower bound heap as a list of buckets.
        The number of non-empty buckets is $O(\log \spread)$.
        If node $x$ is in the $m^{th}$ bucket then $\beta^m < \lb(x) \le \beta^{m+1}$.
        Nodes in the same bucket are accessed in an arbitrary order.
        Local lower bounds may increase over the course of the algorithm.
        So a node $z$ may move to a bucket with a higher priority.
    }
    \label{fig:bucket_queue}
\end{figure}
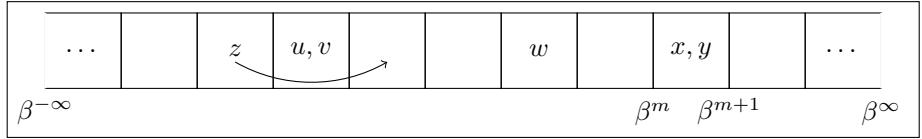

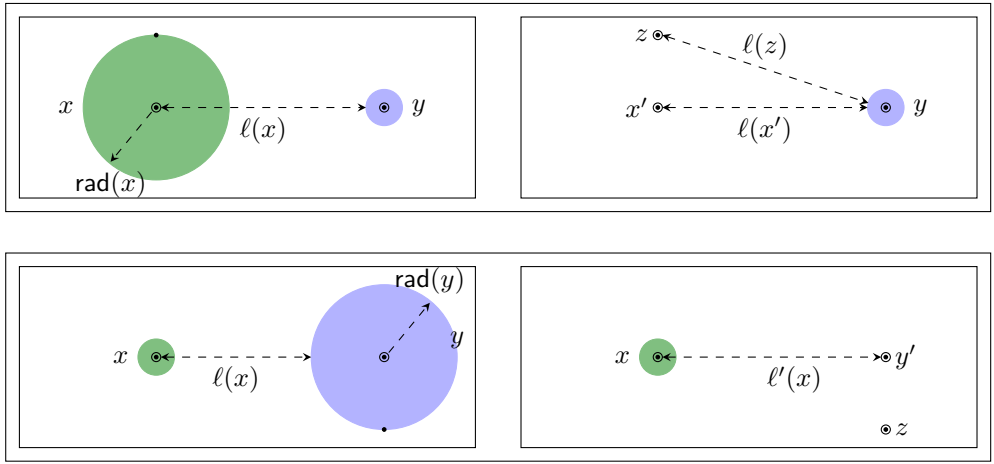
\begin{figure}
    \begin{center}
        \begin{tikzpicture}[>={stealth[scale=3]}, scale=1.2]
            \def\green{black!50!green!50}
            \def\blue{blue!30}
            \draw (-0.15, -0.15) rectangle (10.65, 2.15);
            \draw (0,0) rectangle (5,2);

            \filldraw[\green] (1.5, 1) circle (0.8);
            \filldraw[\blue] (4, 1) circle (0.2);

            \filldraw (1.5, 1) circle (0.02);
            \filldraw (1.5, 1.8) circle (0.02);
            \filldraw (4, 1) circle (0.02);

            \draw (1.5, 1) circle (0.05);
            \draw (4, 1) circle (0.05);

            \draw[<->, dashed] (1.55, 1) -- (3.8, 1) node[midway, below]{$\lb(x)$};
            \draw[->, dashed] (1.45, 0.95) -- (1, 0.4) node[below]{$\noderad{x}$};
            \draw (0.7, 1) node[left]{$x$};
            \draw (4.2, 1) node[right]{$y$};
            
            \def \xshft{5.5}
            \draw (0+\xshft,0) rectangle (5+\xshft,2);

            \filldraw[\blue] (4+\xshft, 1) circle (0.2);

            \filldraw (1.5+\xshft, 1) circle (0.02) node[left]{$x'$};
            \filldraw (1.5+\xshft, 1.8) circle (0.02) node[left]{$z$};
            \filldraw (4+\xshft, 1) circle (0.02);

            \draw (1.5+\xshft, 1) circle (0.05);
            \draw (4+\xshft, 1) circle (0.05);
            \draw (1.5+\xshft, 1.8) circle (0.05);

            \draw (4.2+\xshft, 1) node[right]{$y$};

            \draw[<->, dashed] (1.55+\xshft, 1) -- (3.8+\xshft, 1) node[midway, below]{$\lb(x')$};
            \draw[<->, dashed] (1.55+\xshft, 1.8) -- (3.81+\xshft, 1.05) node[midway, above]{$\lb(z)$};
        \end{tikzpicture}

        \vspace{0.5cm}

        \begin{tikzpicture}[>={stealth[scale=3]}, scale=1.2]
            \def\green{black!50!green!50}
            \def\blue{blue!30}
            \draw (-0.15, -0.15) rectangle (10.65, 2.15);
            \draw (0,0) rectangle (5,2);

            \filldraw[\green] (1.5, 1) circle (0.2);
            \filldraw[\blue] (4, 1) circle (0.8);

            \filldraw (1.5, 1) circle (0.02);
            \filldraw (4, 1) circle (0.02);
            \filldraw (4, 0.2) circle (0.02);

            \draw (1.5, 1) circle (0.05);
            \draw (4, 1) circle (0.05);

            \draw (1.3, 1) node[left]{$x$};
            \draw (4.8, 1) node[above]{$y$};

            \draw[<->, dashed] (1.55, 1) -- (3.2, 1) node[below, midway] {$\lb(x)$};
            \draw[->, dashed] (4.05, 1.05) -- (4.5, 1.6) node[above]{$\noderad{y}$};
            
            \def \xshft{5.5}
            \draw (0+\xshft,0) rectangle (5+\xshft,2);

            \filldraw[\green] (1.5+\xshft, 1) circle (0.2);

            \filldraw (1.5+\xshft, 1) circle (0.02);
            \filldraw (4+\xshft, 1) circle (0.02) node[right]{$y'$};
            \filldraw (4+\xshft, 0.2) circle (0.02) node[right]{$z$};

            \draw (1.3+\xshft, 1) node[left]{$x$};

            \draw (1.5+\xshft, 1) circle (0.05);
            \draw (4+\xshft, 1) circle (0.05);
            \draw (4+\xshft, 0.2) circle (0.05);

            \draw[<->, dashed] (1.55+\xshft, 1) -- (3.95+\xshft, 1);
            \draw (3+\xshft, 1) node[below]{$\lb'(x)$};
        \end{tikzpicture}
    \end{center}
    \caption{
        This figure shows the two cases when local lower bounds may increase.
        At the top, a node $x$ centered at a point in $A$ is replaced by its children, the leaves $x'$ and $z$.
        In this case, $\lb(x') = \lb(x)$ and $\lb(z) \le \lb(x) + \noderad{x}$.
        At the bottom, a node $y$ centered at a point in $B$ is split into its children, the leaves $y'$ and $z$.
        Here, we denote the new lower bound of $\nodectr{x}$ by $\lb'(x)$.
        In this case, $\lb'(x) \le \lb(x) + \noderad{y}$.
    }
    \label{fig:inc_lower_bounds}
\end{figure}

\begin{observation}
    \label{obs:loc_lower_bounds_inc}
    Let $z$ be the next node to be processed with radius $r$.
    \begin{enumerate}
        \item Let $z \in G_B$.
        For a neighbor $x \in N(z)$, the local lower bound of $x$ and all descendents of $x$ centered at $\nodectr{x}$ will never exceed $\ell(x) + r$  over the course of the algorithm.
        \item Let $z \in G_A$.
        Then, the local lower bound of any right child of $z$ will never exceed $\ell(z) + r$ at the end of the iteration.
        Moreover, the local lower bound of any descendent of $z$ will never exceed $\ell(z) + 2r$ over the course of the algorithm.
    \end{enumerate}
\end{observation}

We need to modify the algorithm to avoid revisiting buckets.
We can use the following observation.

\begin{observation}
    \label{obs:same_bucket}
    Let the radius of the next node to be processed be $r$.
    Let $s$ be such that $r \le \frac{\beta-1}{2}\beta^s$.
    Let $x$ be a node such that $\lb(x) \le\beta^m$ for some $m \ge s$.
    Then, the local lower bound of $x$ or any child node of $x$ never exceeds $\beta^{m+1}$.
\end{observation}
\begin{proof}
    It follows from Observation~\ref{obs:loc_lower_bounds_inc} that over the course of the algorithm, the local lower bound of $x$ never exceeds $\lb(x)+r$.
    Moreover, for any child of $x$, the local lower bound never exceeds $\lb(x)+2r$.
    By our choice of $s$ and $m$,
    \[
        \lb(x)+2r \le \beta^m +(\beta-1)\beta^s \le \beta^{m+1}.
    \]
    Therefore, the local lower bounds never exceed $\beta^{m+1}$.
\end{proof}

Now, we describe modified finishing and local lower bound update procedures to use the bucket queue as a lower bound heap.
Instead of finishing nodes, we now finish entire buckets.
Let $s = \left \lceil \log_\beta (\frac{2r\beta}{\beta-1})\right \rceil $.
Each time the radius decreases we update $s$ and traverse the array until we reach the new bucket $s$.
For any occupied bucket $j$ encountered before or coinciding with bucket $s$, append $\beta^j$ to the output for each $a \in \points(x)$ for each node $x$ in bucket $j$.
In other words, we finish all nodes with lower bounds greater than $\beta^s$.
When updating local lower bounds, if there is a node such that its new lower bound would exceed $\beta^s$ then finish it instead of updating its key.

\subsection{Analysis}
The analysis of the \algokHD algorithm closely follows the analysis of the \algoHD algorithm.
As points are removed, the Hausdorff distance decreases, allowing the viability graph to maintain a constant degree throughout as in Theorem~\ref{thm:hd_correct}.
The main difference in the running time analysis is that one must include the cost of the heap operations.

First, we show that the output of \algokHD is correct.
\begin{lemma}\label{lem:algokhd_correct}
    Let $(\delta_0,\dots,\delta_{n-1})$ be the the sequence of distances returned by \algokHD.
    Then, $\delta_i\le \dist_h^{(i)}(A, B) \le (1+\e)\delta_i$ for all $i \in [n]$.
\end{lemma}
\begin{proof}
    Let $L$ be the global lower bound when $\delta_i$ is returned.
    Let $\beta=(1 + \frac{\e}{2})$.
    Then, $L \le \beta \delta_i$.
    Furthermore, $r \le \frac{\beta-1}{2}\delta_i$ by the finishing condition.
    For each $j$, let $p_j$ denote the $j^\text{th}$ removed point.
    Let $S_i = A \backslash \{p_j \mid j \le i\}$.
    It follows by the definition of partial Hausdorff distance and Lemma~\ref{lem:l_plus_2r} that,
    \begin{align*}
        \dist_h^{(i)}(A,B) &\le \dist_h(S_i, B)\\
            &\le L+2r \\
            &\le \beta\delta_i + (\beta-1)\delta_i\\
            &= (1+\e)\delta_i.
    \end{align*}

    For the other direction, it is sufficient to show that $\delta_i\le \dist_h^{(i)}(A, B)$.
    Suppose that $\dist_h^{(i)}(A,B) < \delta_i$.
    Then, there exists some set $S \subset A$ such that $|S| = |S_i|$ and $\dist_h(S,A) = \dist_h^{(i)}(A,B)$.
    As $\dist_h(S,B) < \delta_i$, for all $a\in S$, we have $\dist(a,B)< \delta_i$.
    By our choice of $s$ and $\beta$ the $\delta_i$ must be non-increasing.
    Therefore, none of the points in $S$ could have been removed.
    It follows that $S = S_i$ and this is a contradiction.
    Therefore, $\delta_i\le \dist_h^{(i)}(A,B) \le (1+\e)\delta_i$ for all $i \in [n]$.
\end{proof}

Now we analyze the running time of \algokHD.
\begin{lemma}\label{lem:algokhd_time}
    Let $\beta=(1 + \frac{\e}{2})$.
    \algokHD runs in $\left(2+\frac{1}{\e}\right)^{O(d)} n + O(\log_\beta\spread)$ time.
\end{lemma}
\begin{proof}
    In every iteration, \algokHD performs some operations on the viability graph and some operations on the lower bound heap.
    The degree of a node in the viability graph is $(2+\frac{1}{\e})^{O(d)}$ by the same argument as in Theorem~\ref{thm:hd_correct}.
    By Observation~\ref{obs:same_bucket} and our choice of $s$, all partial distances can be computed in a single traversal of the lower bound heap.
    Therefore, traversing the lower bound heap takes $O(\log_\beta\spread)$ time over the course of the whole algorithm.
    Thus, the total running time is $(2+\frac{1}{\e})^{O(d)}n + O(\log_\beta\spread)$.
\end{proof}

By Lemmas~\ref{lem:algokhd_correct} and~\ref{lem:algokhd_time}, we conclude the following theorem.
\begin{theorem}\label{thm:algokhd}
    Let $\beta=(1 + \frac{\e}{2})$.
    \algokHD computes a $(1+\e)$-approximation of all partial Hausdorff distances in $\left(2+\frac{1}{\e}\right)^{O(d)} n + O(\log_\beta\spread)$ time.
  \end{theorem}


    \section{An Application that Amortizes the Preprocessing Time}\label{sec:mds}
The preprocessing time can be amortized in applications where many Hausdorff distance computations are required.
For example, in metric multi-dimensional scaling (MDS), all pairwise distances are computed among a sets to give a low-dimensional embedding.
\algoHD takes $O(n\log\spread)$ preprocessing time to construct a greedy tree for each input set.
These greedy trees can be re-used in later distance computations.
Thus, for large collections of sets, the preprocessing cost does not impact the asymptotic running time.
As detailed below, this improves on the running time by a factor of $\log n$ over prior methods.

The goal of MDS is to embed a sample of a metric space into a lower dimensional Euclidean subspace.
In this case, the metric space is defined by the Hausdorff distance on subsets of points.
Assume we are given $k$ subsets, $X_1,\dots, X_k$, each of size $n$.
The problem is to compute an MDS embedding of these $k$ sets under the Hausdorff metric.
Then to obtain the input for MDS we need to compute the $\binom{k}{2}$ pairwise Hausdorff distances.
It takes $O(kn\log\spread)$ time to compute the $k$ greedy trees, one for each subset $X_i$.
Each distance is computed in $\e^{-O(d)}n$ time, so the total running time $T$, including preprocessing time, is
\[
    T(n,k) = \underbrace{kn\log \spread}_{\text{greedy tree construction}} + \underbrace{k^2\e^{-O(d)}n}_{\text{distance computations}} = kn(\log \spread + k\e^{-O(d)}).
\]

Thus, for computing the input to MDS, the greedy tree construction time becomes insignificant when $\log \spread$ is $k\e^{-O(d)}$.
In that case, 
\[ T(n, k) = k^2\e^{-O(d)}n.\]

In comparison with results from the literature, Hausdorff distances can easily be computed in $\e^{-O(d)} n\log n$ time using any data structure supporting a $\e^{-O(d)}\log n$ time approximate nearest neighbor search~\cite{krauthgamer04navigating,ram09linear,chubet23proximity,Arya98Optimal}.
Thus, one could compute the input for MDS in $O(k^2n \log n)$ time.
Using greedy trees improves this algorithm by a factor of $\log n$.



    \section{Conclusion}\label{sec:conclusion}
We presented an algorithm to approximate the directed Hausdorff distance that runs in $(2+\frac{1}{\e})^{O(d)}n$ time after computing the greedy trees.
The benefits of preprocessing outweigh the costs if the same sets are involved in multiple distance computations, such as in MDS computations.
With some modifications, the same algorithm can be used to compute all $(1+\e)$-approximate $k$-partial Hausdorff distances in $(2+\frac{1}{\e})^{O(d)}n + O(\log_\beta \spread)$ time.




    \small


\end{document}